\documentclass[Journal]{IEEEtran}
\ifCLASSINFOpdf

\else

\fi

\usepackage{stfloats}
\usepackage{psfig}
\usepackage{amssymb}
\usepackage[cmex10]{amsmath}
\usepackage{graphicx}
\usepackage{epstopdf}
\usepackage{color}
\usepackage{epsfig}
\usepackage{epsf}
\usepackage{cite}
\usepackage{slashbox}
\usepackage{float}
\usepackage{stfloats}
\usepackage{soul}
\usepackage{subfigure}
\usepackage{lineno}
\usepackage{algorithm}
\usepackage{algorithmic}
\usepackage{amsfonts,amssymb}
\usepackage{amsmath,amsthm}
\newtheorem{theorem}{Theorem}

\newtheorem{proposition}{Proposition}
\newtheorem{remark}{Remark}
 \newtheorem{corollary}{Corollary}

\makeatletter

\newcommand{\Rmnum}[1]{\expandafter\@slowromancap\romannumeral #1@}
\makeatother
\begin{document}
\title{\huge A Recursion-Based SNR Determination Method for Short Packet Transmission: Analysis and Applications}

\author{\IEEEauthorblockN{Chengzhe~Yin, Rui~Zhang$^*$,~\IEEEmembership{Member,~IEEE,}
Yongzhao~Li$^*$,~\IEEEmembership{Senior~Member,~IEEE,} Yuhan~Ruan,~\IEEEmembership{Member,~IEEE}, Tao~Li,~\IEEEmembership{Member,~IEEE}, and Jiaheng Lu}

\vspace{-10pt}

\thanks{
Chengzhe Yin, Rui Zhang, Yongzhao Li, Yuhan Ruan, Tao Li, and Jiaheng Lu (graduate student) are with the School
of Telecommunications Engineering, Xidian University, Xi'an 710071, China (E-mail: czyin@stu.xidian.edu.cn; rz@xidian.edu.cn; yzhli@xidian.edu.cn; yhruan@xidian.edu.cn; taoli@xidian.edu.cn; 16010188039@stu.xidian.edu.cn).
}
}
\maketitle

\vspace{-10pt}
\begin{abstract}
The short packet transmission (SPT) has gained much attention in recent years. In SPT, the most significant characteristic is that the finite blocklength code (FBC) is adopted. With FBC, the signal-to-noise ratio (SNR) cannot be expressed as an explicit function with respect to the other transmission parameters. This raises the following two problems for the resource allocation in SPTs: (i) The exact value of the SNR is hard to determine, and (ii) The property of SNR w.r.t. the other parameters is hard to analyze, which hinders the efficient optimization of them. To simultaneously tackle these problems, we have developed a recursion method in our prior work. To emphasize the significance of this method, we further analyze the convergence rate of the recursion method and investigate the property of the recursion function in this paper. Specifically, we first analyze the convergence rate of the recursion method, which indicates it can determine the SNR with low complexity. Then, we analyze the property of the recursion function, which facilitates the optimization of the other parameters during the recursion. Finally, we also enumerate some applications for the recursion method. Simulation results indicate that the recursion method converges faster than the other SNR determination methods. Besides, the results also show that the recursion-based methods can almost achieve the optimal solution of the application cases.

\end{abstract}
\vspace{-5pt}
\begin{IEEEkeywords}
short packet transmission, finite blocklength, SNR determination, convergence analysis, resource allocation
\end{IEEEkeywords}
\IEEEpeerreviewmaketitle
\vspace{-10pt}
\section{Introduction}
Recently, the short packet transmission (SPT) has gained much attention due to the demands of mission-critical tasks~\cite{ref31e}~--\cite{refUAV}. A significant characteristic of SPT is that the finite blocklength code (FBC) is adopted~\cite{ref3}. With FBC, the traditional Shannon's capacity is not applicable since the law of large number is no longer valid. Instead, there exists a backoff from the traditional Shannon's capacity, and the backoff can be characterized by a parameter referred to as channel dispersion~\cite{ref11}.

Owing to the existence of channel dispersion, block error rate (BLER) is no longer negligible, and needs to be taken into account in the expression of FBC achievable rate, which renders the coupling relationship between communication parameters, i.e., BLER, signal-to-noise ratio (SNR), blocklength, and packet size, extremely complicated~\cite{ref11}. Following the FBC achievable rate, the authors in~\cite{ref10b} further pointed out that the SNR cannot be expressed as an explicit function with respect to the other parameters. This raises an interesting SNR determination problem: What is the minimum SNR to meet the certain transmission requirement, i.e., the other three parameters are given, for the SPTs.


To address this problem, the authors in~\cite{p2} approximated the channel dispersion as a constant to make the SNR expression explicit. Unfortunately, when applying this method, the approximation error exists and it increases as the SNR decrease, which makes the method more applicable for the high-SNR case. For any range of SNR, some iteration methods, i.e., fixed point iteration~\cite{fix1} and the bisection method~\cite{ref10b}, were proposed to determine the exact SNR value. Nevertheless, both the above two methods have only linear convergence rate, which can be further improved. Solving the problem from another aspect, the authors in~\cite{p1} derived the analytical solution for the implicit SNR on the basis of the extended lambert $\mathcal{W}$ function. Compared with the methods in~\cite{ref10b} and~\cite{fix1}, the analytical solution is much more intuitive. However, since the analytical solution is the sum of infinite complicated polynomials, it is still not efficient enough for the SNR determination. To determine the exact value of SNR, we also proposed a recursion method in our prior work~\cite{myp1}. In this method, the recursion function is an extremely tight approximation for the implicit SNR function, which gives the method potential to determine the SNR efficiently. However, we have not analyzed the recursion method from this aspect in~\cite{myp1}.

In addition to the SNR determination, the implicity of SNR also raises another problem in resource allocation: The property, e.g., convexity and monotonicity, of SNR w.r.t. the other parameters is hard to analyze. This further hinders the efficient optimization of these parameters to achieve some specific goals, e.g. power consumption minimization~\cite{myp1},\cite{myp}. Fortunately, this problem can also be solved by using the recursion method. The reason is that the recursion function is not only an explicit function that is easier to analyze, but it is also an upper bound approximation of the implicit SNR. This facilitates the optimization of the other transmission parameters during the recursion, on the basis of the majorization-minimization (MM) optimization framework~\cite{mml}\footnote{The other three exact SNR determination methods cannot fulfill the conditions of using MM. Therefore, when using these methods to determine the SNR, the other parameters cannot be optimized on the basis of the MM optimization framework.}.
In this regard, we have verified the packet size can be efficiently optimized during the recursion~\cite{myp1}. It is of interest to further investigate whether other parameters can also be efficiently optimized during the recursion.

Against this background, we further analyze the recursion method in-depth in this paper. Specifically, we first theoretically analyze the convergence rate of the method, which illustrates that the recursion method can determine the SNR with low complexity. Then, to efficiently optimize the other parameters during the recursion, we analyze the property of the recursion function w.r.t. these parameters. The main contributions are as follows.
\begin{itemize}
\item We analyze the convergence rate of the recursion method. Specifically, we first prove that the recursion method achieves quadratic convergence rate for the SNR determination. Then, for the corresponding quadratic convergence factor, we also prove that it is smaller than $1$ for a wide range of SNR, i.e., $[0.25,\infty)$. Combined with the convergence rate analysis, we also illustrate that the recursion method converges faster than the other SNR determination methods by simulation.
\item
   We analyze monotonicity and convexity of the recursion function w.r.t. the packet size and BLER. Specifically, it is proved that the recursion function is strictly monotony and convex w.r.t. both of them, separately. Besides, the joint convexity of the recursion function w.r.t. these two parameters is also verified for the typical SPT configurations. Combined with the above property analysis, we also enumerate some applications for the recursion method.
\end{itemize}

The remainder of this paper is organized as follows. In Section II, we analyze the capacity of the recursion method in SNR determination from the aspect of convergence rate. In Section III, we investigate the optimization of the other parameters during the recursion. Finally, Section IV concludes this paper.

\vspace{-10pt}
\section{The Recursion Method for SNR determination}
In this section, we first state the SNR determination problem and introduce the recursion method developed in our prior work. Then, we further analyze the convergence rate of the recursion method, which indicates that it can determine the SNR with low complexity.

\vspace{-10pt}
\subsection{The SNR Determination Problem}
According to~\cite{ref11}, the normal approximation of the upper bound achievable rate under FBC is given by
\begin{equation}
R = \frac{N}{m} \approx {\log _2}(1 + \gamma ) - \sqrt {\frac{V}{m}} \frac{{{Q^{ - 1}}(\varepsilon )}}{{\ln 2}}, \label{dss}
\end{equation}
where ${N}$ is the packet size, ${m}$ is the blocklength, $\gamma$ is the SNR, ${V}$ is the channel dispersion given by ${V = 1 - 1/{{{{(1 + \gamma )}^2}}}}$, ${Q( \cdot )}$ is the Gaussian Q-function, and ${\varepsilon }$ is the BLER. The results in~\cite{ref11} indicate that the normal approximation is accurate when $m>\ddot m$, where $\ddot m\approx 20$.



According to~\eqref{dss}, $\gamma$ cannot be directly expressed as an explicit function w.r.t. the other three parameters. For ease of expression, we denote the coupling relationship between $\gamma$ and the other three parameters by the following implicit function, i.e., $\gamma = {\Gamma (N,m,\varepsilon )}$.
According to Theorem~1 in~\cite{ref10b}, ${\Gamma (N,m,\varepsilon )}$ is continuous and differentiable. Besides, owing to the monotonicity of  $N$~\cite[Proposition~1]{ref10b}, $m$~\cite[Proposition~1]{ref37}, and $\varepsilon$~\cite[Proposition~1]{ref37o} w.r.t. $\gamma$, $\Gamma (N,m,\varepsilon )$ is unique for any given $N>0$, $m>0$, and $0<\varepsilon<0.5$.

The SNR determination problem is to determine the value of ${\Gamma (N,m,\varepsilon )}$ when the other three parameters are given.
\vspace{-10pt}
\subsection{The Recursion Method}

Inspired by the approximation of channel dispersion developed in~\cite{refhr}, we have designed a recursion method to determine ${\Gamma (N,m,\varepsilon )}$ in our prior work~\cite{myp1}. In this method, the recursion function, which is referred to as exponential approximation recursion (EAR) function, is given by
\begin{equation}
\dot \gamma^{(j)} \!=\! \tilde \Gamma(N,m,\varepsilon,\dot \gamma^{(j-1)})\triangleq {\exp\left(\frac{{\frac{N}{m}\ln 2 + \mu (\dot \gamma^{(j-1)} )b}}{{1 - \rho (\dot \gamma^{(j-1)} )b}}\right)} - 1, \label{ap}
\end{equation}
where $x\triangleq y$ means $x$ is defined by $y$, and $\dot \gamma^{(j)}$ denotes the recursion solution for $\gamma$ in $j$-th round. In~\eqref{ap}, coefficients $q$ and $b$ are defined as $q\triangleq Q^{-1}(\varepsilon)$ and $b\triangleq q/\sqrt m$, respectively; Functions $\rho(\dot \gamma)$ and $\mu(\dot \gamma)$ are defined as follows
\begin{equation}
\rho(\dot \gamma)  \triangleq \frac{1}{(1 +\dot   \gamma){\sqrt {{{\dot  \gamma}^2} + 2\dot  \gamma} }}, \label{app2}
\end{equation}
and
\begin{equation}
\mu(\dot  \gamma)  \triangleq \sqrt {1 - \frac{1}{{{{(1 + \dot  \gamma)}^2}}}} - \frac{\ln (1 + \dot  \gamma)}{{(1 + \dot  \gamma)\sqrt {{{\dot  \gamma}^2} + 2\dot  \gamma} }}. \label{app3}
\end{equation}

Note that ${1 - \rho (\dot \gamma )b}$, $\rho(\dot \gamma)$, and $\mu(\dot \gamma)$ are all positive for any feasible $\dot \gamma$. As for the first term, its positivity can be checked as follows
\begin{equation}
\begin{split}
1 - \frac{b}{{\left( {\dot \gamma  + 1} \right)\sqrt {{{\dot \gamma }^2} + 2\dot \gamma } }}& \ge 1 - \frac{b}{{\left( {\bar \gamma  + 1} \right)\sqrt {{{\bar \gamma }^2} + 2\bar \gamma } }}\\
&\mathop  = \limits^{(a)} 1 - \frac{{\ln (1 + \bar \gamma )}}{{{{\bar \gamma }^2} + 2\bar \gamma }}\mathop  \ge \limits^{(b)} 0, \label{5a}
\end{split}
\end{equation}
where ${\bar \gamma }$ is the threshold $\gamma$ such that $0 = \ln (1 + \bar \gamma ) - \frac{{\sqrt {{{\bar \gamma }^2} + 2\bar \gamma } }}{{\bar \gamma  + 1}}b$, which yields (a); (b) holds true for $\bar \gamma \geq0$. Note that $\dot \gamma \geq \bar \gamma$ is also the feasible region for ($\dot \gamma$)\footnote{An error was made in~\cite{myp1} that $\dot \gamma$ must be feasible, i.e., $\dot \gamma \ge \bar \gamma$, to guarantee that the EAR function is an upper bound of ${\Gamma (N,m,\varepsilon )}$. This can be easily proved by checking the derivative of $\tilde \Gamma(\dot \gamma)$.}. As for the last term, we have
\begin{equation}
\mu (\dot\gamma )\mathop  \ge \limits^{(a)} \mu (\bar \gamma )\mathop  \ge \limits^{(b)} 0,
\end{equation}
where (a) holds since $\mu(\dot\gamma)$ is strictly increasing (this can be verified by analyzing its first-order derivative); (b) can be proved as similar as~\eqref{5a}.

The recursion method starts from any feasible point that is greater than ${\Gamma (N,m,\varepsilon )}$, e.g., $\dot\gamma^{(0)} = \hat \gamma \triangleq \exp \left( {\frac{N}{m}\ln 2 + b} \right) -1$, and finally converges to ${\Gamma (N,m,\varepsilon )}$~\cite[Theorem~2]{myp1}.
\vspace{-10pt}
\subsection{Convergence Rate Analysis}
In this subsection, we further investigate the convergence rate of the recursion method by analyzing its convergence factor. For simplicity, we denote $\Gamma(N,m,\varepsilon)$ by $\gamma$ in this subsection.


The convergence factor $\mathcal{{\cal Q}}_p$ is defined as follows~\cite{com1}
\begin{equation}
{\mathcal{{\cal Q}}_p} = \mathop {\lim }\limits_{t \to \infty } \frac{{\left| {{y^ * } - {y_{t + 1}}} \right|}}{{{{\left| {{y^ * } - {y_t}} \right|}^p}}}, \label{qp}
\end{equation}
where $p$ corresponds to the convergence order, $y^ *$ denotes the solution of the problem, and $y_{t}$ denotes the solution in $t$-th round of iteration (or recursion).

If $0<\mathcal{{\cal Q}}_1 <1$, then the algorithm has linear convergence rate; If $\mathcal{{\cal Q}}_1 = 0$, then the algorithm has superlinear convergence rate; If $0<\mathcal{{\cal Q}}_2 <\infty$, then the algorithm has quadratic convergence rate. On this basis, we analyze the convergence rate of the recursion method in the following theorem.
\vspace{-5pt}
\begin{theorem}
The recursion method has quadratic convergence rate, and the corresponding convergence factor is given by
\begin{equation}
g_1 \triangleq \frac{{2{\gamma ^2} + 4\gamma  + 1}}{{2{{(1 + \gamma )}^2}\left( {{\gamma ^2} + 2\gamma } \right)\sqrt {{\gamma ^2} + 2\gamma } \left( {1 - \frac{b}{{(1 + \gamma )\sqrt {{\gamma ^2} + 2\gamma } }}} \right)}}b .
\end{equation}
\end{theorem}
\vspace{-10pt}
\begin{proof}
We first illustrate that the recursion method achieves superlinear convergence rate in~\eqref{1}, which is on the top of next page. In~\eqref{1}, (a) holds owing to~\cite[Theorem~2]{myp1}; (b) holds since the term $\exp( {\frac{{\frac{N}{m}\ln 2 - \ln (1 + \dot \gamma ) + b\sqrt {\dot V} }}{{1 - \rho (\dot \gamma )b}}}) \buildrel \Delta \over = {J(\dot \gamma )}$ is replaced by its first-order Taylor expansion at $\dot \gamma=\gamma$; (c) and (d) hold following from~(1). Then, we further prove that the recursion method has quadratic convergence rate in~\eqref{1s}, which is on the top of next page. In~\eqref{1s}, (a) holds since the term ${J(\dot \gamma )}$ is replaced by its second-order Taylor expansion at $\dot \gamma=\gamma$; (b) can be derived as similar as~\eqref{1}. Following from~\eqref{1s}, it is easy to observe that ${\mathcal{{\cal Q}}_2}$ is bounded for this method if ${1 - \rho (\gamma )b}>0$, which can be proved as similar as~\eqref{5a}. This completes the proof.
\newcounter{mytempeqncnt7}
\begin{figure*}[!t]
\normalsize

\begin{equation}
\begin{split}
&\mathop {\lim }\limits_{k \to \infty } \frac{{{{\dot \gamma }^{(k + 1)}} - \gamma }}{{{{\dot \gamma }^{(k)}} - \gamma }} \mathop  = \limits^{(a)} \mathop {\lim }\limits_{\dot \gamma  \to \gamma } \frac{{\exp \left( {\frac{{\frac{N}{m}\ln 2 + \mu (\dot \gamma )b}}{{1 - \rho (\dot \gamma )b}}} \right) - \gamma  - 1}}{{\dot \gamma  - \gamma }} \mathop   = \mathop {\lim }\limits_{\dot \gamma  \to \gamma } \frac{{J(\dot \gamma )\left( {1 + \dot \gamma} \right) - \gamma  - 1}}{{\left( {\dot \gamma  - \gamma } \right)}}\\
& \mathop  = \limits^{(b)}  \mathop {\lim }\limits_{\dot \gamma  \to \gamma } \frac{{\left( {{J(\gamma )} + {{\left. {{{ {{J^\prime(\dot \gamma )}} } }} \right|}_{\dot \gamma  = \gamma }}\left( {\dot \gamma  - \gamma } \right) + o(\dot \gamma  - \gamma )} \right)\left( {1 + \dot \gamma } \right) - \gamma  - 1}}{{\left( {\dot \gamma  - \gamma } \right)}} \mathop  = \limits^{(c)} 1 +  {\left. {{J^\prime(\dot \gamma ) }} \right|_{_{\dot \gamma  = \gamma }}}\left( {1 + \gamma } \right) \\
&= 1 + \frac{{\left( { - \frac{1}{{1 + \gamma }} + \frac{b}{{{{\left( {1 + \gamma } \right)}^2}\sqrt {{\gamma ^2} + 2\gamma } }}} \right)\left( {1 - \rho (\gamma )b} \right) + \rho '(\gamma )b\left( {\frac{N}{m}\ln 2 - \ln (1 + \gamma ) + b\sqrt V } \right)}}{{{{\left( {1 - \rho (\gamma )b} \right)}^2}}}\left( {1 + \gamma } \right)  \mathop = \limits^{(d)}1 - \frac{{{{\left( {1 - \rho (\gamma )b} \right)}^2}}}{{{{\left( {1 - \rho (\gamma )b} \right)}^2}}} = 0. \label{1}
\end{split}
\end{equation}
\hrulefill
\normalsize
\begin{equation}
\begin{split}
&\mathop {\lim }\limits_{\dot \gamma  \to \gamma } \frac{{\exp \left( {\frac{{\frac{N}{m}\ln 2 + \mu (\dot \gamma )b}}{{1 - \rho (\dot \gamma )b}}} \right) - \gamma  - 1}}{{{{\left( {\dot \gamma  - \gamma } \right)}^2}}} \! \mathop  =  \limits^{(a)} \! \mathop {\lim }\limits_{\dot \gamma  \to \gamma } \frac{{\left( {J(\gamma ) \!+\! {{\left. {J'(\dot \gamma )} \right|}_{\dot \gamma  = \gamma }}\left( {\dot \gamma  - \gamma } \right) \!+\! {{\left. {J''(\dot \gamma )} \right|}_{\dot \gamma  = \gamma }}\frac{{{{\left( {\dot \gamma  - \gamma } \right)}^2}}}{2} + o{{(\dot \gamma  - \gamma )}^2}} \right)\left( {1 + \dot \gamma } \right) - \gamma  - 1}}{{{{\left( {\dot \gamma  - \gamma } \right)}^2}}}\\
&  \mathop  = \limits^{(b)} {{{\left. {J''(\dot \gamma )} \right|}_{\dot \gamma  = \gamma }}}\frac{{1 + \gamma }}{2} - \frac{1}{{1 + \gamma }} + \mathop {\lim }\limits_{\dot \gamma  \to \gamma } \frac{{o{{(\dot \gamma  - \gamma )}^2}\left( {1 + \dot \gamma } \right)}}{{\left( {\dot \gamma  - \gamma } \right)^2}}
= \frac{{2{\gamma ^2} + 4\gamma  + 1}}{{2{{(1 + \gamma )}^2}\left( {{\gamma ^2} + 2\gamma } \right)\sqrt {{\gamma ^2} + 2\gamma } \left( {1 - \frac{b}{{(1 + \gamma )\sqrt {{\gamma ^2} + 2\gamma } }}} \right)}}b. \label{1s}
\end{split}
\end{equation}
\hrulefill
\vspace{-15pt}
\end{figure*}
\end{proof}
\vspace{-5pt}

In what follows, we further investigate the relationship between $g_1$ and $\gamma$. Since the monotonicity of $g_1$ w.r.t. $\gamma$ is hard to check, we analyze the condition when $g_1$ is smaller than a typical value, i.e., $g_1 \leq 1$, in the following corollary.


\vspace{-5pt}
\begin{corollary}
The quadratic convergence factor of the recursion method is smaller than $1$ if
\vspace{-5pt}
\begin{equation}
\frac{{2{\gamma ^2} + 4\gamma  + 1}}{{2(1 + \gamma )\left( {{\gamma ^2} + 2\gamma } \right)}} + 1 - \frac{{{\gamma ^2} + 2\gamma }}{{\ln (1 + \gamma )}} \le 0, \label{9o}
\end{equation}
and the solution of~\eqref{9o} is given by $\gamma  \ge  0.25$.
\end{corollary}
\vspace{-10pt}
\begin{proof}
The inequality $g_1\leq1$ is equivalent to the following inequality
\vspace{-5pt}
\begin{equation}
b \le \frac{{2{{(1 + \gamma )}^2}\left( {{\gamma ^2} + 2\gamma } \right)\sqrt {{\gamma ^2} + 2\gamma } }}{{2(1 + \gamma )\left( {{\gamma ^2} + 2\gamma } \right) + \left( {2{\gamma ^2} + 4\gamma  + 1} \right)}}. \label{eq9}
\end{equation}

Owing to~\eqref{dss}, we have $\frac{N\ln2}{m} = \ln (1 + \gamma ) - \sqrt V b \ge 0$. Therefore, it holds true that $\frac{{\sqrt {{\gamma ^2} + 2\gamma } }}{{(1 + \gamma )\ln (1 + \gamma )}} \le \frac{1}{b}$. On this basis, if the following inequality holds,
\begin{equation}
\begin{split}
&\frac{{2{\gamma ^2} + 4\gamma  + 1}}{{2{{(1 + \gamma )}^2}\left( {{\gamma ^2} + 2\gamma } \right)\sqrt {{\gamma ^2} + 2\gamma } }} + \frac{1}{{(1 + \gamma )\sqrt {{\gamma ^2} + 2\gamma } }} \\
&\le \frac{{\sqrt {{\gamma ^2} + 2\gamma } }}{{(1 + \gamma )\ln (1 + \gamma )}},
\end{split}
\end{equation}
then~\eqref{eq9} must hold. From the above inequality,~\eqref{9o} can be easily derived. Besides, it can be proved that the left hand side of~\eqref{9o} is strictly decreasing w.r.t. $\gamma$. Therefore, one can work out the solution of~\eqref{9o}, which is $\gamma  \ge  0.25$.
\end{proof}
\vspace{-5pt}
In Fig.~\ref{RV}, we illustrate the search error (the error between the iteration value and the real value) versus iteration number and flops (floating point operations) for the four exact SNR determination methods, i.e., the bisection method~\cite{ref10b}, the fixed point iteration~\cite{fix1}, the analytical solution~\cite{p1}, and the recursion method. It shows that the recursion method converges faster than the other three methods\footnote{It is well known that the bisection method has only linear convergence rate. Besides, relying on the Taylor expansion, one can also prove that the fixed point iteration method proposed in~\cite{fix1} has linear convergence rate. Moreover, since the analytical solution is not obtained by iteration, we only compare it in Fig.~1(b). It is worth noting that the operations, including $+$,$-$,$*$,$/$,$\sqrt{(\cdot)}$,$\cdot!$,$\exp( \cdot)$, $\ln(\cdot)$, and $Q^{-1}(\cdot)$, are counted as a single flop. Besides, the elements that occur multiple times are counted only once, e.g., $N*\ln(2)/m$ occurs in each round, but it is counted only once ($3$ flops).}.


\begin{figure}[htbp]
	\centering
	\subfigure[Search error vs. iteration number] {\includegraphics[width=.4\textwidth]{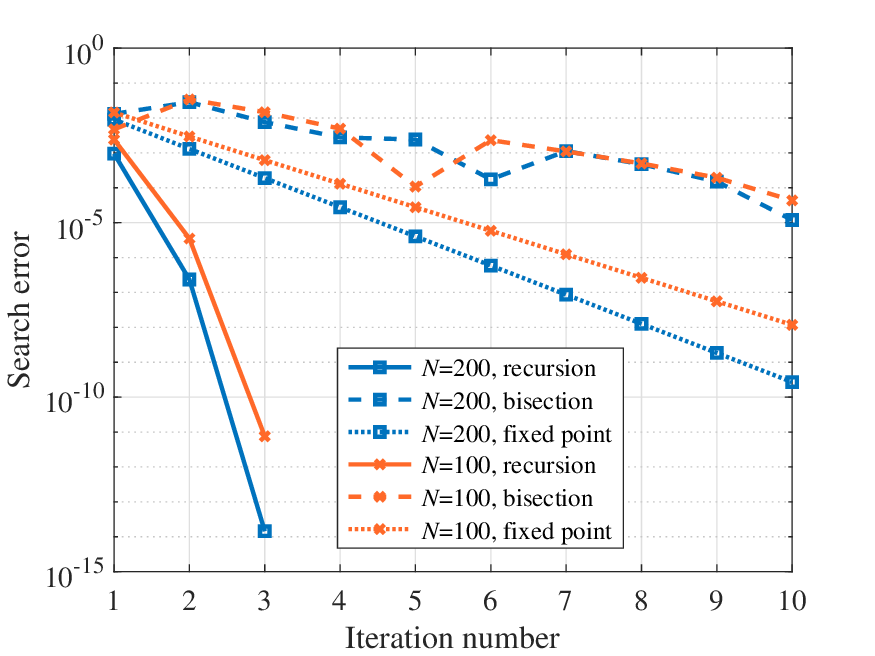}}
	\subfigure[Search error vs. flops]  {\includegraphics[width=.4\textwidth]{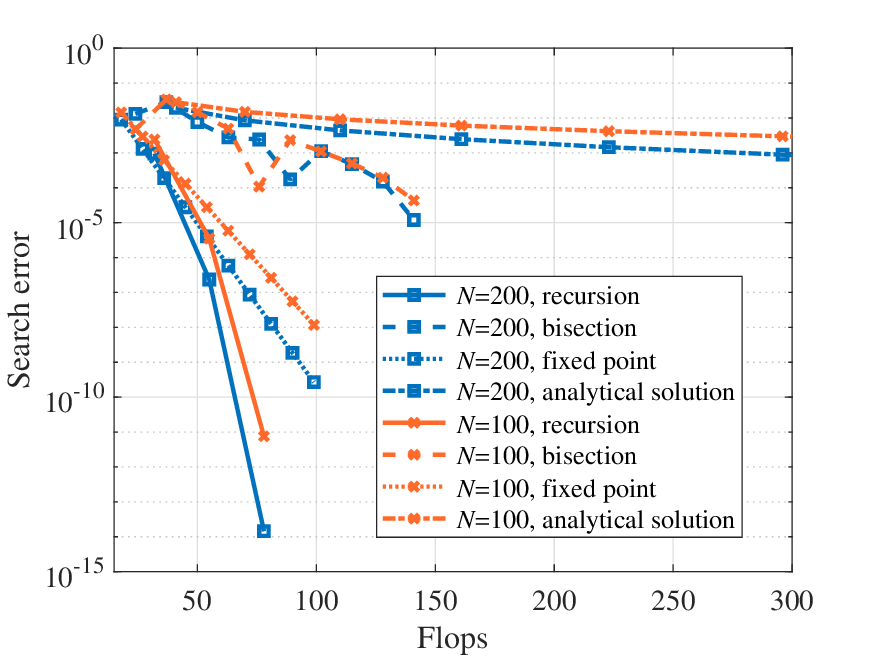}}
\vspace{-10pt}
	\caption{Search error of different methods when $m=1000$ and $\varepsilon = 10^{-5}$.}
\vspace{-10pt}
	\label{RV}
\end{figure}


\begin{remark}
The convergence rate analysis is also correlated to the complexity of the method. Specifically, Theorem~1 indicates that if $\dot \gamma_t \to \gamma$, then $\left| {\gamma - {\dot\gamma_{t + 1}}} \right| =g_1{\left| {\gamma - {\dot\gamma_t}} \right|^2}$. Besides, if $\dot \gamma_t  \ne \gamma$, then there must exist a $\delta>0$ such that $\left| {\gamma - {\dot\gamma_{t + 1}}} \right| \le \delta{\left| {\gamma - {\dot\gamma_t}} \right|^2}$ holds true for any $t$. Combining the above two cases, we have $\left| {\gamma - {\dot\gamma_{t + 1}}} \right| \le \max\{g_1,\delta\}{\left| {\gamma - {\dot\gamma_t}} \right|^2}$. From the above inequality, it is easy to obtain that there exists a $k$ such that for any $t\ge k$, the complexity of the remaining steps is $\mathcal O(\log\log \frac{1}{\omega  })$, where $\omega$ is the accuracy~\cite{com1}. In other words, if the initial point $\dot \gamma_0$ is close to $\gamma$, then the complexity of the overall recursion method is $\mathcal O(\log\log \frac{1}{\omega  })$.
\end{remark}

\vspace{-10pt}
\section{Parameter Optimization During the Recursion}
In the last section, we force on the SNR determination problem in which the other parameters are given and only the SNR needs to be determined. In this section, we further consider the problems in which both the parameters and the SNR have to be optimized to achieve some specific goals, e.g.,
\begin{equation}
\begin{split}
&\mathop {{\min} }\limits_{A,\gamma} f_1(A,\gamma)\\
&\; {\rm s.t.}\;\,f_2(A,\gamma) \le 0\\
&\;\;\;\;\;\;\;\,\gamma=\Gamma(A),
\label{pr1}
\end{split}
\end{equation}
where $A \in \{ N,m,\varepsilon \}$, and $f_1$ and $f_2$ denotes the objective function and the inequality constraint, respectively. For ease of analysis, we hereby assume that $f_1$ and $f_2$ are the linear combinations of $A$ and $\gamma$. Besides, the combination coefficient of the term $\gamma$ is assumed to be non-negative.

Different from the SNR determination problem, the determination of SNR is not the only challenge in~\eqref{pr1}. The implicity of $\Gamma (A)$ also renders the convexity analysis of the overall problem difficult. To simultaneously tackle these problems, an efficient way is to optimize the parameters during the recursion of SNR. Specifically, in $j$-th recursion, the SNR is still updated as per~\eqref{ap}. But instead, the updating here should take into account the optimization of the other parameters, i.e., $\dot \gamma^{(j)} = \tilde \Gamma(A^{(j-1)},\dot \gamma^{(j-1)})$. Then, with the updated $\dot \gamma^{(j)}$, $A^{(j)}$ can be updated by solving the following sub-problem,
\begin{equation}
\begin{split}
&\mathop {{\min} }\limits_A f_1(A,\tilde \Gamma (A,\dot \gamma^{(j)}))\\
&\; {\rm s.t.}\;\,f_2(A,\tilde \Gamma (A,\dot \gamma^{(j)})) \le 0. \label{pr3}
\end{split}
\end{equation}
By repeating these two steps, a Karush-Kuhn-Tucker point of~\eqref{pr1} can be obtained until convergence~\cite[Theorem~3]{myp1}.

In the above procedure, a curial problem is whether~\eqref{pr3} can be efficiently solved, which is contingent upon the property of the EAR function $\tilde \Gamma$. Therefore, in this section, we further analyze the property of the EAR function. Combined with the property analysis, we also enumerate some applications for the recursion method and show the simulation results.

\subsection{Property Analysis}
To efficiently solve the transformed sub-problems, we analyze the monotonicity and the convexity of the EAR function in the following proposition.
\begin{proposition}
It holds true that

1. The EAR function is strictly monotony w.r.t. $N$ and $\varepsilon$.

2. The EAR function is convex w.r.t. $N$ and $\varepsilon$, respectively.

3. The EAR function is jointly convex w.r.t. $N$ and $\varepsilon$ if the following inequality holds true,
\begin{equation}
\sqrt m  \le \frac{{q\sqrt {{{\gamma^*}}^2 + 2{{\gamma^*}}} }}{{(1 + {{\gamma^*}})\ln (1 + {{\gamma^*}})}},
\end{equation}
where $\gamma^*$ is given by
\begin{equation}
{q^2} = \frac{{\ln (1 + \gamma^*)}}{{\left( {{{\gamma^*}^2} + 2\gamma^*} \right)\left( {\left( {{{\gamma^*}^2} + 2\gamma^*} \right) - \ln (1 + \gamma^*)} \right)}}.
\end{equation}
\end{proposition}

\begin{proof}
We first prove Proposition 1.1 as follows. The first-order derivatives of the EAR function are given as follows
\begin{equation}
\frac{{\partial \tilde \Gamma}}{{\partial N}} = K(\dot \gamma )\frac{{\frac{1}{m}\ln 2}}{{1 - \rho (\dot \gamma )b}},
\end{equation}
and
\begin{equation}
\frac{{\partial \tilde \Gamma}}{{\partial \varepsilon }} =  - K(\dot \gamma )\frac{{\mu (\dot \gamma ) + \frac{N}{m}\ln 2\rho (\dot \gamma )}}{{{{\left( {1 - \rho (\dot \gamma )b} \right)}^2}}}\frac{{\sqrt {2\pi } {e^{\frac{{{q^2}}}{2}}}}}{{\sqrt m }},
\end{equation}
where  $\exp \left( {\frac{{\frac{N}{m}\ln 2 + \mu (\dot \gamma )b}}{{1 - \rho (\dot \gamma )b}}} \right) \buildrel \Delta \over = K(\dot \gamma )$ for simplicity. Owing to the positivity of $\mu (\dot \gamma)$, $\rho (\dot \gamma)$, and $1-\rho (\dot \gamma)b$, the above two partial derivatives indicate that the EAR function is monotony w.r.t. $N$ and $\varepsilon$, which completes the proof of Proposition 1.1.

Then, Proposition 1.2 is proved as follows. The second-order derivatives of the EAR function are given as follows
\begin{equation}
\frac{{{\partial ^2}\tilde \Gamma}}{{\partial {N^2}}} = K(\dot \gamma ){\left( {\frac{{\frac{1}{m}\ln 2}}{{1 - \rho (\dot \gamma )b}}} \right)^2},
\end{equation}
and
\begin{equation}
\begin{split}
&\frac{{{\partial ^2}\tilde \Gamma}}{{\partial {\varepsilon ^2}}} =K(\dot \gamma ) \frac{{q2\pi {e^{{q^2}}}}}{{\sqrt m }}\frac{{\mu (\dot \gamma ) + \frac{N}{m}\ln 2\rho (\dot \gamma )}}{{{{\left( {1 - \rho (\dot \gamma )b} \right)}^2}}}\\
&+ K(\dot \gamma ){\left( {\frac{{\left( {\mu (\dot \gamma ) + \frac{N}{m}\ln 2\rho (\dot \gamma )} \right)}}{{{{\left( {1 - \rho (\dot \gamma )b} \right)}^2}}}\frac{{\sqrt {2\pi } {e^{\frac{{{q^2}}}{2}}}}}{{\sqrt m }}} \right)^2}\\
&+K(\dot \gamma ) \frac{{2\pi {e^{{q^2}}}}}{m}\frac{{2\rho (\dot \gamma )\left( {\mu (\dot \gamma ) + \frac{N}{m}\ln 2\rho (\dot \gamma )} \right)}}{{{{\left( {1 - \rho (\dot \gamma )b} \right)}^3}}},
\end{split}
\end{equation}
and
\begin{equation}
\begin{split}
&\frac{{\partial^2 {\tilde \Gamma}}}{{\partial \varepsilon \partial N}} =  - K(\dot \gamma )\frac{{\frac{1}{m}\ln 2\rho (\dot \gamma )}}{{{{\left( {1 - \rho (\dot \gamma )b} \right)}^2}}}\frac{{\sqrt {2\pi } {e^{\frac{{{q^2}}}{2}}}}}{{\sqrt m }}\\
& - K(\dot \gamma )\frac{{\mu (\dot \gamma ) + \frac{N}{m}\ln 2\rho (\dot \gamma )}}{{{{\left( {1 - \rho (\dot \gamma )b} \right)}^3}}}\frac{{\sqrt {2\pi } {e^{\frac{{{q^2}}}{2}}}\ln 2}}{{m\sqrt m }}.
\end{split}
\end{equation}

\newcounter{mytempeqncnt8}
\begin{figure*}[!t]
\normalsize
\vspace{-5pt}
\begin{align}
&\frac{{{\partial ^2}\tilde \Gamma}}{{\partial {N^2}}}\frac{{{\partial ^2}\tilde \Gamma}}{{\partial {\varepsilon ^2}}} - {\left( {\frac{{\partial ^2}\tilde \Gamma}{{\partial \varepsilon \partial N}}} \right)^2} \buildrel\textstyle.\over= q\sqrt m \left( {\mu (\dot \gamma ) + \frac{N}{m}\rho (\dot \gamma )\ln 2} \right) - {\rho ^2}(\dot \gamma ) \label{4}\\
&= q\sqrt m \left( { - \frac{1}{{\sqrt {{{\dot \gamma}^2} + 2\dot \gamma} }}\frac{1}{{(1 + \dot \gamma)}}\ln (1 + \dot \gamma) + \sqrt {1 - \frac{1}{{{{(1 + \dot \gamma)}^2}}}}  + \frac{{\frac{N}{m}\ln 2}}{{(1 + \dot \gamma)\sqrt {{{\dot \gamma}^2} + 2\dot \gamma} }}} \right) - \frac{1}{{{{(1 + \dot \gamma)}^2}\left( {{{\dot \gamma}^2} + 2\dot \gamma} \right)}} \tag{\ref{4}{a}} \label{4a}\\
&\ge \frac{{{q^2}\sqrt {{{\dot \gamma}^2} + 2\dot \gamma} }}{{\ln (1 + \dot \gamma)(1 + \dot \gamma)}}\left( { - \frac{1}{{\sqrt {{{\dot \gamma}^2} + 2\dot \gamma} }}\frac{1}{{(1 + \dot \gamma)}}\ln (1 + \dot \gamma) + \sqrt {1 - \frac{1}{{{{(1 + \dot \gamma)}^2}}}} } \right) - \frac{1}{{{{(1 + \dot \gamma)}^2}\left( {{{\dot \gamma}^2} + 2\dot \gamma} \right)}} \triangleq g_2(\dot \gamma) \tag{\ref{4}{b}} \label{4b}.
\end{align}
\hrulefill
\vspace{-15pt}
\end{figure*}
\vspace{-5pt}
Since $\mu (\dot \gamma)$, $\rho (\dot \gamma)$, and $1-\rho (\dot \gamma)b$ are all positive for $\dot \gamma>0$, both $\frac{{{\partial ^2}\tilde \Gamma}}{{\partial {N^2}}}$ and $\frac{{{\partial ^2}\tilde \Gamma}}{{\partial {\varepsilon ^2}}}$ are positive, which completes the proof of Proposition 1.2.

To prove Proposition 1.3, we first derive the lower bound of $\frac{{{\partial ^2}\tilde \Gamma}}{{\partial {N^2}}}\frac{{{\partial ^2}\tilde \Gamma}}{{\partial {\varepsilon ^2}}} - {\left( {\frac{{\partial ^2}\tilde \Gamma}{{\partial \varepsilon \partial N}}} \right)^2}$ in~\eqref{4}, which is on the top of next page. In~\eqref{4}, $A\buildrel\textstyle.\over= B$ denotes $A$ and $B$ have the same sign;~\eqref{4b} holds for the following thee reasons. First, we have $\sqrt m = \frac{{q\sqrt {{{\bar \gamma }^2} + 2\bar \gamma } }}{{\left( {1 + \bar \gamma } \right)\ln (1 + \bar \gamma )}}$. Second, it can be proved that $\frac{{\sqrt {{{ x}^2} + 2 x} }}{{\ln (1 +  x)(1 +  x)}}$ is strictly decreasing w.r.t. $x$ for $x>0$. Third, we have $\dot \gamma \geq \bar \gamma$ and $N>0$.

To prove $g_2(\dot \gamma)\geq0$, the following inequalities should hold,
\begin{equation}
\begin{split}
&{q^2}\left( {\frac{{\left( {{{\dot \gamma}^2} + 2\dot \gamma} \right) - \ln (1 + \dot \gamma)}}{{\ln (1 + \dot \gamma){{(1 + \dot \gamma)}^2}}}} \right) \ge \frac{1}{{{{(1 + \dot \gamma)}^2}\left( {{{\dot \gamma}^2} + 2\dot \gamma} \right)}}\\
&  \Leftrightarrow {q^2} \ge \frac{{\ln (1 + \dot \gamma)}}{{\left( {{{\dot \gamma}^2} + 2\dot \gamma} \right)\left( {\left( {{{\dot \gamma}^2} + 2\dot \gamma} \right) - \ln (1 + \dot \gamma)} \right)}}\triangleq g_3(\dot \gamma).
\end{split}
\end{equation}
It follows from the above that the EAR function is jointly convex w.r.t. $N$ and $\varepsilon$ if $q^2 \geq g_3(\dot \gamma)$. Then, we investigate the property of $g_3(\dot \gamma)$ as follows. The first-order derivative of $g_3(\dot \gamma)$ is given by
\begin{equation}
\begin{split}
&g_3'(\dot \gamma) \buildrel\textstyle.\over= \left( {{{\dot \gamma}^2} + 2\dot \gamma} \right)\left( {\left( {{{\dot \gamma}^2} + 2\dot \gamma} \right) - \ln (1 + \dot \gamma)} \right)\\
& - 2{\left( {\dot \gamma + 1} \right)^2}\left( {\left( {{{\dot \gamma}^2} + 2\dot \gamma} \right) - \ln (1 + \dot \gamma)} \right)\ln (1 + \dot \gamma)\\
& - \left( {2{{\left( {\dot \gamma + 1} \right)}^2} - 1} \right)\left( {{{\dot \gamma}^2} + 2\dot \gamma} \right)\ln (1 + \dot \gamma)\\
& = 2{\left( {\dot \gamma + 1} \right)^2}\ln (1 + \dot \gamma)\left( {\underbrace {\ln (1 + \dot \gamma) - \left( {{{\dot \gamma}^2} + 2\dot \gamma} \right)}_{{g_4(\dot \gamma)} < 0}} \right)\\
& + \left( {{{\dot \gamma}^2} + 2\dot \gamma} \right)\left( {\underbrace {\left( {{{\dot \gamma}^2} + 2\dot \gamma} \right) - 2{{\left( {\dot \gamma + 1} \right)}^2}\ln (1 + \dot \gamma)}_{{g_5(\dot \gamma)} < 0}} \right),
\end{split}
\end{equation}
which indicates that $g_3(\dot \gamma)$ is strictly decreasing, where the negativity of $g_4(\dot \gamma)$ and $g_5(\dot \gamma)$ can be checked by analyzing their first-order derivatives. For simplicity, we denote by $ {\gamma^*}$ the solution of $q^2 = g_3(\dot \gamma)$. If $\dot \gamma \geq {\gamma^*}$, then the EAR function is jointly convex w.r.t. $N$ and $\varepsilon$.

To ensure $\dot \gamma \geq {\gamma^*}$, one can further impose additional constraints on $m$ based on the following inequality
\begin{equation}
\sqrt m  \le \frac{{q\sqrt {{{\gamma^*}}^2 + 2{{\gamma^*}}} }}{{\ln (1 + {{\gamma^*}})(1 + {{\gamma^*}})}}, \label{21t}
\end{equation}
where the reasons are described as follows. As similar as the derivation from~\eqref{4a} to~\eqref{4b}, we have $\frac{{q\sqrt {{{\dot \gamma}^2} + 2\dot \gamma} }}{{\ln (1 + \dot \gamma)(1 + \dot \gamma)}} \le \sqrt m $. Combining the above inequality and~\eqref{21t}, we have
\begin{equation}
\frac{{q\sqrt {{{\dot \gamma}}^2 + 2{{\dot \gamma}}} }}{{\ln (1 + {{\dot \gamma}})(1 + {{\dot \gamma}})}} \le \sqrt m  \le \frac{{q\sqrt {{{\gamma^*}}^2 + 2{{\gamma^*}}} }}{{\ln (1 + {{\gamma^*}})(1 + {{\gamma^*}})}}, \label{22t}
\end{equation}
Since $\frac{{\sqrt {{{ x}^2} + 2 x} }}{{\ln (1 +  x)(1 +  x)}}$ is strictly decreasing w.r.t. $x$ for $x>0$, $\dot \gamma \geq  {\gamma^*}$ can be derived from~\eqref{22t}. This completes the proof of Proposition 1.3.
\end{proof}

Based on Proposition~1.3, we then discuss the convexity of the EAR function for some typical SPT configurations. According to the current standard~\cite{ref31e}, the BLER threshold is normally set to $\varepsilon= 10^{-5}$, yielding $q\approx4.26$ and ${\gamma^*}\approx 0.025$. In this case,~\eqref{21t} is equivalent to $\sqrt m\leq 37.8705$, which is normally satisfied in SPT. Note that $\frac{{\sqrt {{{\gamma^*}}^2 + 2{{\gamma^*}}} }}{{\ln (1 + {{\gamma^*}})(1 + {{\gamma^*}})}}$ is strictly decreasing w.r.t. $\gamma^*$, and $\gamma^*$ is strictly decreasing w.r.t. $q$. Therefore, both $\frac{{\sqrt {{{\gamma^*}}^2 + 2{{\gamma^*}}} }}{{\ln (1 + {{\gamma^*}})(1 + {{\gamma^*}})}}$ and $\frac{{q\sqrt {{{\gamma^*}}^2 + 2{{\gamma^*}}} }}{{\ln (1 + {{\gamma^*}})(1 + {{\gamma^*}})}}$ are strictly increasing w.r.t. $q$. It means that if the BLER requirement is more stringent, e.g. another typical BLER requirement in the current standard is $\varepsilon \leq 10^{-9}$ ($q\approx6$)~\cite{ref31e}, then~\eqref{21t} will be further relaxed, and more likely to be satisfied in SPT.
\vspace{-10pt}
\subsection{Applications}
In the pervious subsection, we have proved that the EAR function is convex w.r.t. $N$ and $\varepsilon$. In this subsection, we further enumerate some applications of using the recursion method to optimize them. For simplicity, we denote $\tilde \Gamma(N_i,m_i,\varepsilon_i,\dot \gamma_i)$ by $\tilde \Gamma_i$ in this subsection.


\subsubsection{Weighted Sum Rate Maximization for Multi-User Orthogonal Multiple Access}

We consider a multi-user orthogonal multiple access case in which the packets are delivered to multiple users, orthogonally. In accordance with~\eqref{pr3}, the corresponding weighted sum rate maximization problem can be formulated as follows
\begin{equation}
\begin{split}
&\mathop {\max }\limits_{{N_i}} \sum {{\alpha _i}{N_i}}\\
&\;{\rm{s}}.{\rm{t}}.\,\;{\mkern 1mu} \sum {\frac{{m_i{{\tilde \Gamma }_i}}}{{{h_i}}}}  \le {P_{\max }},
\end{split}
\end{equation}
where $B_i,B \in \{ \alpha ,N,m,h,\varepsilon \}$ denotes the corresponding $B$ in $i$-th orthogonal transmission, $\alpha$ denotes the weighted coefficient, $\varepsilon_i = \varepsilon_{\rm th}$ is the BLER threshold, and ${P_{\max }}$ denotes the maximum transmit power threshold at transmitter. Besides, $h$ denotes the normalized channel gain (normalized by additive complex white Gaussian noise). Therefore, $m_i{\tilde \Gamma }_i/h_i$ denotes the consumed power for $i$-th orthogonal transmission.

According to Proposition~1.2, the convexity of the above problem can be easily checked. Therefore, the above problem can be efficiently solved by the convex optimization algorithms, such as the prime-dual inner-point method~\cite{math2}.

\subsubsection{Power Consumption Minimization for Multi-Hop Relaying}
We consider a typical multi-hop relay transmission case, in which the source node packet is delivered to the destination node via a series of relay nodes. Assuming that blocklength and transmission size are given, i.e., $N_i=N$ and $m_i = m$, the following power consumption minimization problem can be formulated as per~\eqref{pr3},
\begin{equation}
\begin{split}
&\mathop {\min }\limits_{{\varepsilon _i}} \sum {\frac{{m{{\tilde \Gamma }_i}}}{{{h_i}}}}   \\
&\;{\rm{s}}.{\rm{t}}.{\mkern 1mu} \;\prod {(1 - {\varepsilon _i}) \ge 1 - {\varepsilon _{{\rm{th}}}}} ,
\end{split}
\end{equation}
where $C_i,C \in \{ N,m,h,\varepsilon \}$ denotes the corresponding $C$ in $i$-th hop transmission.

By approximating $\prod {(1 - {\varepsilon _i})} $ as $1 - \sum {{\varepsilon _i}}$~\cite{p2}, the above problem can be converted to a convex problem, where the convexity of the objective function can be easily checked according to Proposition~1.2. Therefore, the above problem can be efficiently solved as well.

\subsubsection{Energy Efficiency Maximization for Two-Hop Relaying}
A similar problem to the power consumption minimization problem is the energy efficiency (EE) maximization problem with the spectrum efficiency (SE) constraint. In a two-hop relay transmission, the EE maximization problem can be formulated as follows
\begin{equation}
\begin{split}
&\mathop {\max }\limits_{{N},{\varepsilon _1}} \frac{N}{ {\frac{{m{{\tilde \Gamma }(N,m,\varepsilon_1,\dot \gamma_1)}}}{{{h_1}}}}+{\frac{{m{{\tilde \Gamma }(N,m,\varepsilon_{\rm th}-\varepsilon _1,\dot \gamma_2)}}}{{{h_2}}}}}\\
&\;{\rm{s}}.{\rm{t}}.\;\frac{N}{{m} } \ge {\phi _{{\rm{th}}}}, \label{24f}
\end{split}
\end{equation}
where ${\phi _{{\rm{th}}}}$ is the threshold SE and $\varepsilon _2$ is replaced by $\varepsilon_{\rm th}-\varepsilon _1$ according to the approximate reliability condition\footnote{It is easy to prove that the condition holds with equality at the optimal solution, relying on Proposition 1.1.}, i.e., $1-(1 - {\varepsilon _1})(1 - {\varepsilon _2}) \approx 1 - {{\varepsilon _1}}-\varepsilon _2\ge 1 - {\varepsilon _{{\rm{th}}}}$. According to Proposition~1.3, the above problem is a concave-convex fractional programming problem for typical SPT configurations, e.g., $\sqrt m <37.8705$ and $\varepsilon_{\rm th}<10^{-5}$. Therefore, it can be converted to a convex problem via the Dinkelbach's transform~\cite{math5}, and solved by the gradient-based methods.

\begin{remark}
In the above three applications, the transmissions are all orthogonal. It is also worth noting that in some particular non-orthogonal transmissions, e.g., power domain non-orthogonal multiple access, the transmit power of each user can be expressed as the product of SNR. Since the SNR is approximated as an exp-function in the recursion method, the product of SNR is still an exp-function, which renders the problem still tractable. Please refer to~\cite{myp1} for an example of this.
\end{remark}
\vspace{-10pt}

\subsection{Simulation Results}
To further demonstrate the superiority of the recursion method, we further compare the solutions of applying the proposed recursion method with the optimal solution of the three application cases in Figs.~\ref{1}--\ref{3}. The results show that for all the three application cases, our proposed recursion based algorithm can almost achieve the optimal solution.

 \begin{figure}
\vspace{-12pt}
\centering
     \includegraphics[width=0.45\textwidth]{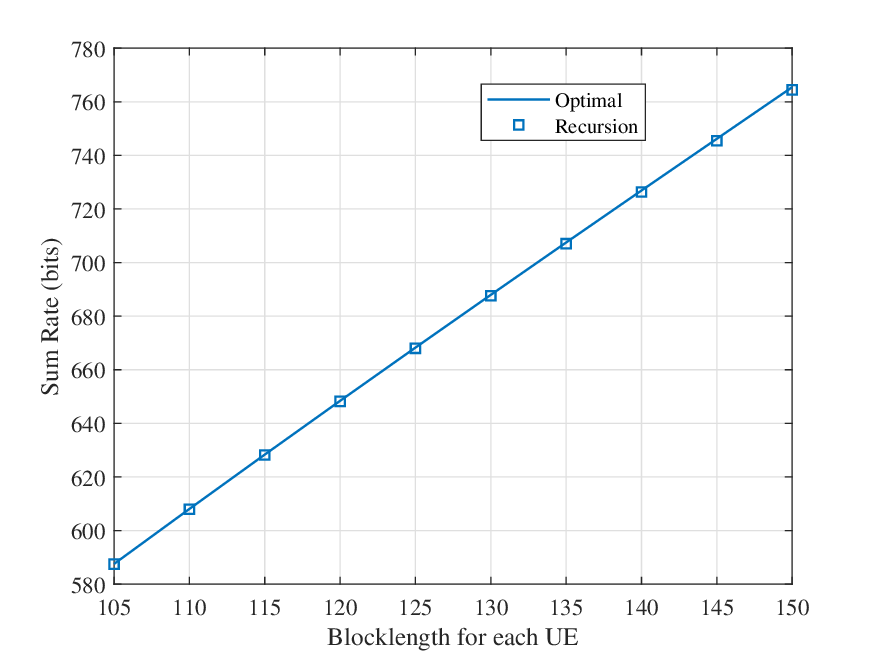}
     \vspace{-12pt}
     \caption{Performance comparisons between the recursion method and exhaustive search for weighted sum rate application, when $\varepsilon_i = {10^{ - 5}}$, $P_{\rm max}=0.0002$~W, and ${\alpha_i}=1$.}

     \label{1}
\end{figure}

 \begin{figure}
\vspace{-12pt}
\centering
     \includegraphics[width=0.45\textwidth]{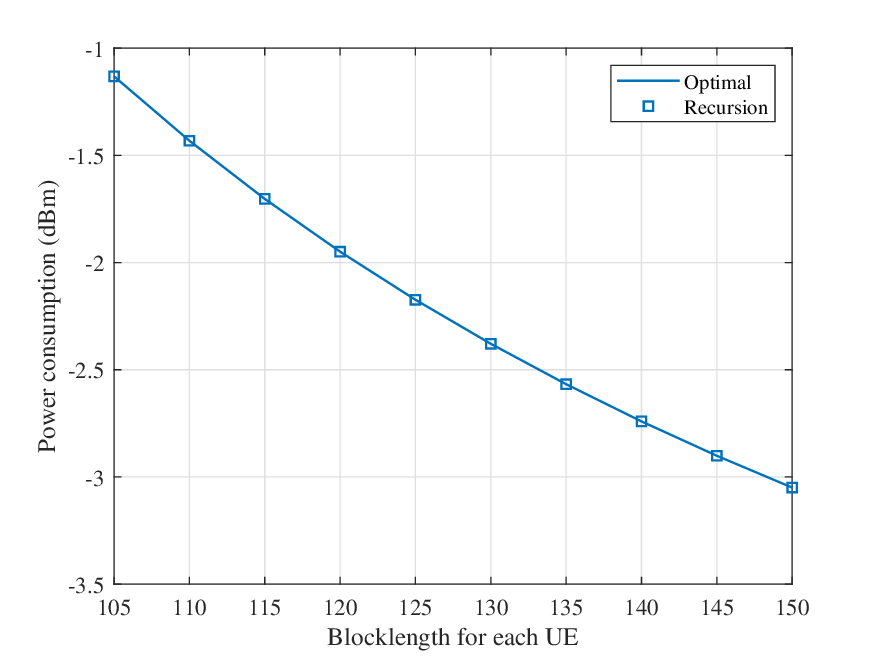}
     \vspace{-12pt}
     \caption{Performance comparisons between the recursion method and exhaustive search for power consumption minimization application, when $N_1=N_2=320$~bits and ${\varepsilon _{\rm th}} = {10^{ - 5}}$.}

     \label{2}
\end{figure}

 \begin{figure}
\vspace{-12pt}
\centering
     \includegraphics[width=0.45\textwidth]{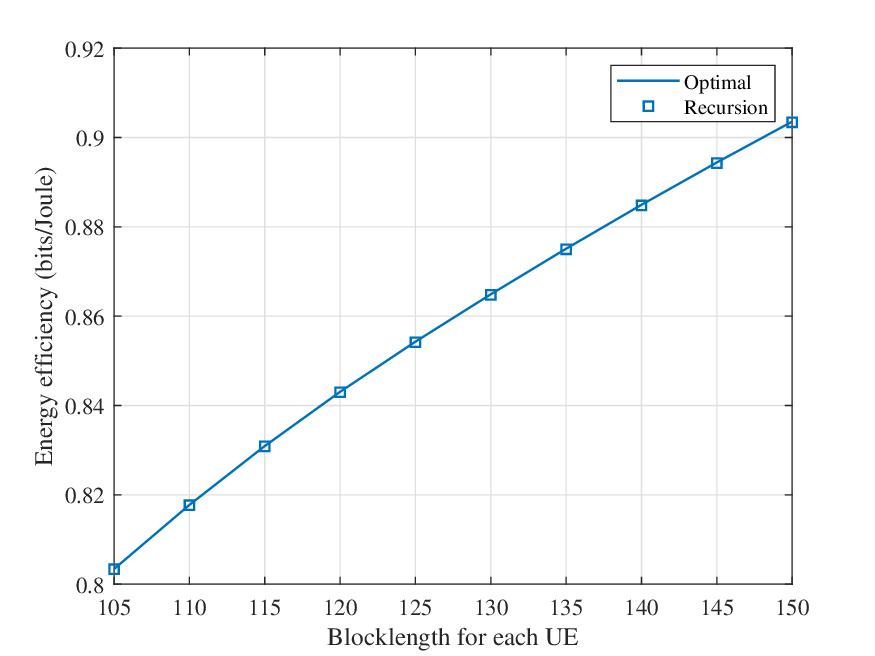}
     \vspace{-12pt}
     \caption{Performance comparisons between the optimization algorithms and exhaustive search for energy efficiency maximization application, when ${\varepsilon _{\rm th}} = {10^{ - 5}}$.}

     \label{3}
\end{figure}
The simulation scenario and the remaining setups are as follows. We consider a simplified downlink $2$-UE scenario for all the three applications. For weighted sum rate application, the two UEs are $20$ and $80$~m away from base station, respectively; For the other two applications, the relay UE is $20$ and $80$~m away from base station and the other UE, respectively. It is assumed that the channel gain is only determined by the path loss to fairly compare different methods. The path loss model is ${PL_{\rm dB}{\rm{ = 32}}{\rm{.4 + 23}}{\rm{lo}}{{\rm{g}}_{10}}(d) + 23{\log _{10}}({f_{\rm{c}}})}$, where ${d}$ is the distance and $f_{\rm{c}}=6$~GHz is the carrier frequency. The frequency bandwidth of each sub-carrier is set to $60$ kHz and the noise power spectral density is set to $- 174$ dBm/Hz. On this basis, the normalized channel gain $h_i$ can be obtained.
\section{Conclusion}

In this paper, we have introduced an exponential-approximation-based recursion method for determining the SNR in SPT. Then, we have proved that the recursion method has quadratic convergence rate. Furthermore, we have proved that the EAR function is monotony and jointly convex w.r.t. the packet size and BLER for typical SPT configurations. Finally, we have enumerated some applications for the recursion method. Simulation results showed that the recursion method converges faster than the other SNR determination methods. Besides, the results also showed that the recursion-based methods can almost achieve the optimal solution of the application cases.

\end{document}